\documentclass[oneside,english]{amsart}
\usepackage[T1]{fontenc}
\usepackage[latin9]{inputenc}
\usepackage[letterpaper]{geometry}
\usepackage{amsthm}
\usepackage{amstext}
\usepackage{graphicx}
\usepackage{amssymb}

\makeatletter

\providecommand{\tabularnewline}{\\}

\numberwithin{equation}{section}
\numberwithin{figure}{section}
\theoremstyle{plain}
\newtheorem{thm}{Theorem}
  \theoremstyle{plain}
  \newtheorem{lem}[thm]{Lemma}

\makeatother

\makeatother

\usepackage{babel}

\makeatother

\usepackage{babel}

\makeatother

\usepackage{babel}

\makeatother

\usepackage{babel}

\makeatother

\usepackage{babel}

\makeatother

\usepackage{babel}

\begin{document}

\title{a cuckoo hashing variant with improved memory utilization and insertion
time}

\author{Ely Porat \and Bar Shalem }

\maketitle

\section{abstract}

Cuckoo hashing \cite{cuckoo hashing} is a multiple choice hashing
scheme in which each item can be placed in multiple locations, and
collisions are resolved by moving items to their alternative locations.
In the classical implementation of two-way cuckoo hashing, the memory
is partitioned into contiguous disjoint fixed-size buckets. Each item
is hashed to two buckets, and may be stored in any of the positions
within those buckets. Ref. \cite{3.5 way} analyzed a variation in
which the buckets are contiguous and overlap. However, many systems
retrieve data from secondary storage in same-size blocks called pages.
Fetching a page is a relatively expensive process; but once a page
is fetched, its contents can be accessed orders of magnitude faster.
We utilize this property of memory retrieval, presenting a variant
of cuckoo hashing incorporating the following constraint: each bucket
must be fully contained in a single page, but buckets are not necessarily
contiguous. Empirical results show that this modification increases
memory utilization and decreases the number of iterations required
to insert an item. If each item is hashed to two buckets of capacity
two, the page size is 8, and each bucket is fully contained in a single
page, the memory utilization equals 89.71\% in the classical contiguous
disjoint bucket variant, 93.78\% in the contiguous overlapping bucket
variant, and increases to 97.46\% in our new non-contiguous bucket
variant. When the memory utilization is 92\% and we use breadth first
search to look for a vacant position, the number of iterations required
to insert a new item is dramatically reduced from 545 in the contiguous
overlapping buckets variant to 52 in our new non-contiguous bucket
variant. In addition to the empirical results, we present a theoretical
lower bound on the memory utilization of our variation as a function
of the page size.

\section{Introduction}

Cuckoo hashing \cite{cuckoo hashing} is a multiple choice hashing
scheme in which each item can be placed in multiple locations, and
collisions are resolved by moving items to their alternative locations.
This hashing scheme resembles the cuckoo's nesting habits: the cuckoo
lays its eggs in other birds' nests. When the cuckoo chick hatches,
it pushes the other eggs out of the nest. Hence the name {}``cuckoo
hashing.'' As Ref. \cite{3.5 way} explains, analysis of hashing
is similar to the analysis of balls and bins. Hashing an item to a
memory location corresponds to throwing a ball into a bin. Insights
from balls and bins processes led to breakthroughs in hashing methods.
For example, if we throw $n$ balls into $n$ bins independently and
uniformly, it is highly probable that the largest bin will get $\left(1+\textrm{o}\left(1\right)\right)\textrm{log}\left(n\right)/\textrm{log}\textrm{log}\left(n\right)$
balls. Azar et. al \cite{Azar} found that if each ball selects two
bins independently and uniformly, and is placed in the bin with fewer
balls, the final distribution is much more uniform. This led to hashing
each item to one of two possible buckets, decreasing the load on the
most-loaded bucket to $\textrm{log\ensuremath{\left(\textrm{log}\left(n\right)\right)}}+\textrm{O}(1)$
with high probability. In general, if each item is hashed into $d\geq2$
buckets, the maximum load decreases to $\textrm{log\ensuremath{\left(\textrm{log}\left(n\right)\right)}}/\textrm{log}\left(d\right)+\textrm{O}(1).$

Cuckoo hashing \cite{cuckoo hashing} is an extension of two-way hashing.
Each item is hashed to a few possible buckets, and existing items
may be moved to their alternate buckets in order to free space for
a new item. There are many variants of cuckoo hashing. The goals of
cuckoo hashing are to increase memory utilization (the number of items
that can be successfully hashed to a given memory size) and to decrease
insertion complexity. Pagh and Rodler \cite{cuckoo hashing} analyzed
hashing of each item to $d=2$ buckets of capacity $k=1$, and demonstrated
that moving items during inserts results in 50\% space utilization
with high probability. Fotakis et. al. \cite{Space efficient hash tables}
analyzed hashing of each item into more than two buckets. Ref. \cite{Efficient hashing with lookups in two memory accesses}
analyzed a practically-important case in which each item is hashed
to $d=2$ buckets of capacity $k=2$. Refs. \cite{The k-orientability thresholds,The random graph threshold for k-orientiability}
found tight memory utilization thresholds for $d=2$ buckets of any
size $k\ge2$. Specifically, they proved that the memory utilization
for $d=2$ and $k=2$ is 89.7\%.

Ref. \cite{Tight Thresholds for Cuckoo Hashing} proved that the maximum
memory utilization thresholds for $d\ge3$ and $k=1$ are equal to
the previously known thresholds for the random k-XORSAT problem. Ref.
\cite{Maximum Matchings,Orientability of Random Hypergraphs} developed
a tight formula for memory utilization for any $d\ge3$ and $k=1$
and Ref. \cite{Multiple-orientability Thresholds} extended the formula
to any $d\ge3$ and $k\ge1$.

\emph{comment added.} While this work was being completed, we became
aware of Ref. \cite{Cuckoo Hashing with Pages} which proposed a model
where the memory is divided into pages and each key has several possible
locations on a single page as well as additional choices on a second
backup page. They provide interesting experimental results.

In a classical implementation of two-way cuckoo hashing, the memory
is partitioned into contiguous disjoint fixed-sized buckets of size
$k$ . Each item is hashed to 2 buckets and may be stored in any of
the $2k$ locations within those buckets. Ref. \cite{3.5 way} analyze
a variation in which the buckets overlap. For example, if the bucket
capacity $k$ is 3, the disjoint bucket memory locations are: $\left\{ 0,1,2\right\} ,\left\{ 3,4,5\right\} ,\left\{ 6,7,8\right\} ,\ldots$.
whereas the overlapping bucket memory locations are: $\left\{ 0,1,2\right\} ,\left\{ 1,2,3\right\} ,\left\{ 3,4,5\right\} ,\ldots$.

Their empirical results show that this variation increases memory
utilization from 89.7\% to 96.5\% for $d=2$ and $k=2$. However,
many systems retrieve data from secondary storage in same-size blocks
called pages. Fetching a page is a relatively expensive process, but
once a page is fetched, its contents can be accessed orders of magnitude
more quickly. We utilize this property of memory retrieval to present
a variant of cuckoo hashing requiring that each bucket be fully contained
in a single page but buckets are not necessarily contiguous.

In this paper we compare the following three variants of cuckoo hashing:
\begin{enumerate}
\item CUCKOO-CHOOSE-K- the algorithm introduced in this paper. The buckets
are \emph{any} $k$ cells in a page, not necessarily in contiguous
locations. There are $\binom{t}{k}$ buckets in a page, where $t$
is the size of the page. 
\item CUCKOO-OVERLAP \cite{3.5 way}- The buckets are contiguous and overlap.
Here we assume that all buckets are fully contained in a single page,
so there are $t-k+1$ buckets in a page. This is a generalization
of Ref. \cite{3.5 way}. Originally Ref. \cite{3.5 way} did not consider
dividing the memory into pages. 
\item CUCKOO-DISJOINT \cite{cuckoo hashing}- The buckets are contiguous
and not overlapping. There are $t/k$ buckets in a page. This is a
generalization of Ref. \cite{cuckoo hashing}. Originally Ref. \cite{cuckoo hashing}
did not consider larger buckets. 
\end{enumerate}
Note that algorithm CUCKOO-DISJOINT is the extreme case of the CUCKOO-OVERLAP
and CUCKOO-CHOOSE\_K algorithms when the size of the page $t$ equals
the size of the bucket $k$.

We prove theoretically and present empirical evidence that our CUCKOO-CHOOSE-K
modification increases memory utilization. Moreover, using the classical
cuckoo hashing scheme, an item insertion requires multiple look-ups
of candidates to displace. Empirical results show that our modification
dramatically decreases the number of candidate look-ups required to
insert an item compared to Ref. \cite{3.5 way}. In the overlapping
buckets variant \cite{3.5 way}, some buckets are split between two
pages, so that each item resides in up to $2d$ pages. In our variant,
each bucket is fully contained in a single page.

An appealing experimental result is that CUCKOO-CHOOSE-K memory utilization
converges very quickly as a function of the page size $t$. When $k=2$
and $t=16$, memory utilization is 0.9763. This value is almost identical
to memory utilization when $t=2^{20}$ which equals 0.9767. CUCKOO-OVERLAP
memory utilization when $t=16$ is 0.9494, and CUCKOO-DISJOINT memory
utilization is only 0.8970. If we allow a tiny gap of one inside the
buckets $(t=3),$ memory utilization increases from 0.9229 (CUCKOO-OVERLAP)
to 0.9480 (CUCKOO-CHOOSE-K). Table \ref{fig:Number-of-iteration}
specifies the parameters used in our analysis.

\section{Theoretical analysis}

We can determine the success of cuckoo hashing by analyzing the cuckoo
hyper graph. The vertices of the graph are the memory locations. The
hyper-edges of the graph connect all the memory locations where each
item could be placed. Recall that each item can be placed in $d$
buckets chosen uniformly and independently of other items. Each bucket
is composed of any $k$ locations in a page.

It is well known (see, e.g., ref. \cite{3.5 way} for a proof) that
a cuckoo hash \emph{fails}$ $ if and only if there is a sub-graph
$S$ with $v$ vertices and more than $v$ edges. We say that a sub-graph
$S$ has \emph{failed} if it has more edges than vertices.

\begin{table}
\begin{tabular}{|l|l|l|}
\hline 
Symbol & Description  & Comments\tabularnewline
\hline
\hline 
$n$  & number of vertices  & $n\rightarrow\infty$\tabularnewline
 & (hash table capacity)  & \tabularnewline
\hline 
$m$  & number of edges  & $m\leq n$, $m\rightarrow\infty$ \tabularnewline
 & (hashed items)  & \tabularnewline
\hline 
$d$  & number of buckets each  & typical value: 2, $dk>2$\tabularnewline
 & item is hashed to  & \tabularnewline
\hline 
$k$  & bucket size  & typical value: 2-3, $dk>2$\tabularnewline
\hline 
$t$  & page size  & $t\geq k$. $k$ divides $t$. $t$ divides $n$.\tabularnewline
\hline 
$g$  & number of pages  & $g=n/t$. \tabularnewline
\hline 
$\beta$  & memory utilization  & $\beta=\frac{m}{n}$,$0<\beta\leq1$$ $\tabularnewline
\hline 
$V$  & a set of vertices  & $\left|V\right|=v$\tabularnewline
\hline 
$S=S\mbox{(V,E) }$  & a sub-graph  & \tabularnewline
\hline 
$v$  & $v=\left|V\right|$  & $dk\le v<m$. \tabularnewline
 &  & (if $v\geq m$ or $v<dk$ then $S$ cannot fail).\tabularnewline
\hline 
$x$  & $x=\frac{v}{n}$  & $\frac{dk}{n}\le x<\beta$$ $.\tabularnewline
\hline 
$\epsilon$ & a small constant & $0<\epsilon\ll1$\tabularnewline
\hline 
$\delta$ & a small constant & $0<\delta\ll1$\tabularnewline
\hline
\end{tabular}\caption{\label{tab:Parameter-names-and}Parameter names and descriptions}

\end{table}

We will begin by analyzing the probability of success of CUCKOO-CHOOSE-K
for the case where the page size $t$ equals the array size $n$.
This simple and special case is presented here to introduce the main
ideas applied in the following section, where we analyze the general
case where the page size is a finite constant (independent of $n$).

\subsection{Memory utilization when the page size $t$ equals the array size
$n$}

An analysis of memory utilization has been performed previously in
\cite{Space efficient hash tables}. In their analysis, they assume
$k=1$ and prove that if $d\ge\frac{2}{\beta}\log\left(\frac{\textrm{e}\beta}{1-\beta}\right)$,
then the hashing will be successful with a probability of at least
$1-\textrm{O}(n^{4-2d})$. Here we derive a similar constraint on
memory utilization $\beta$. We solve the constraint numerically for
different values of $k$ and $d$, and obtain a lower bound on possible
memory utilization for the specified values. We perform the analysis
using a modification of the method in \cite{Space efficient hash tables},
which we will later generalize to page sizes being equal to any given
constant.

We will bound the failure probability using the union bound. But first,
we would like to reduce redundant summations. We observe that: 
\begin{itemize}
\item If there exists a sub-graph $S\left(V,E\right)$ with $\left|E\right|>\left|V\right|$
and there exists an edge that has exactly one vertex $v_{0}$ outside
of V, then the sub-graph $S'\left(V'=V\cup\left\{ v_{0}\right\} ,E'\right)$
also has more edges than vertices since $\left|E'\right|\geq\left|E\right|+1>\left|V\right|+1=\left|V'\right|$. 
\item If there exists a sub-graph $S\left(V,E\right)$, such that $\left|V\right|=v$
and $\left|E\right|>v+1$ then there exists a sub-graph $S'\left(V',E'\right)$
such that $\left|E'\right|=\left|V'\right|+1$. We can find such a
sub-graph simply by adding vertices to $V$ one by one until we get
a sub-graph where the number of edges equals the number of vertices
plus one. 
\end{itemize}
For each sub-graph $S$ we define an indicator variable $Z_{S}.$

$Z_{S}=\begin{cases}
\textrm{1} & S\left(V,E\right)\textrm{\,\ has\,}\left|V\right|=v\textrm{\,\ vertices\,\ and\,}\left|E\right|=v+1\,\textrm{edges}\textrm{\,\ AND}\\
 & \textrm{There\,\ is\,\ no\,\ edge\,\ that\,\ connects\,}V\textrm{\,\ to\,\ exactly\,\ one\,\ vertex\,\ from\,\ outside\,\ of\,}V\\
\textrm{0} & \textrm{otherwise}\end{cases}$

If $Z_{s}=1$, then we will say that $S$ is a \emph{bad} sub-graph,
and otherwise we will say that $S$ is a \emph{good} sub-graph. If
the sum over $Z_{s}$ of all sub-graphs is equal to zero then every
sub-graph is good then the cuckoo hash \emph{succeeded}$ $. We will
find the the memory utilization $\beta$ such that the sum over $Z_{s}$
of all sub-graphs is $\textrm{o}(1)$ as $n\rightarrow\infty$. 

Let $p_{hit}$ be the probability that a random edge hits $V$. Let
$p_{1}$ be the probability that a random edge connects $V$ to exactly
one vertex from outside of V and let $p_{bad}\left(v\right)$ be probability
that \emph{a given }sub-graph $S\left(V,E\right)$ is bad. 
\begin{lem}
$p_{bad}\left(v\right)=\binom{m}{v+1}p_{hit}^{v+1}\left(1-p_{1}-p_{hit}\right)^{m-\left(v+1\right)}$\end{lem}
\begin{proof}
Immediate from the definition of $Z_{s}$. For S to be bad, exactly
$v+1$ edges out of $m$ edges must hit $V$ and all the rest must
miss $V$ and must not connect $V$ to exactly one vertex from outside
of V.\end{proof}
\begin{lem}
The probability that a random edge hits $V$ is $p_{hit}\left(v\right)=\left(\frac{\binom{v}{k}}{\binom{n}{k}}\right)^{d}$\end{lem}
\begin{proof}
Each item is hashed independently to $d$ buckets and the size of
each bucket is $k$. The number of buckets in $V$ is therefore $\binom{v}{k}$
and the total number of buckets is $\binom{n}{k}$.\end{proof}
\begin{lem}
The probability that a random edge connects $V$ to exactly one vertex
from outside of V is $p_{1}=d\frac{\left(n-v\right)\binom{v}{k-1}}{\binom{n}{k}}\left(\frac{\binom{v}{k}}{\binom{n}{k}}\right)^{d-1}$\end{lem}
\begin{proof}
$d-1$ buckets must all fall in $V$ and one bucket must contain any
$k-1$ vertices from $V$ and any of the $n-v$ vertices from outside
of $V$. 
\end{proof}
Let $P_{bad}\left(v\right)$ be the probability that there exists
a sub-graph $S$ with $v$ vertices such that $S$ is bad. According
to the union bound, $P_{bad}\left(v\right)<N\left(v\right)\cdot p_{bad}\left(v\right)$,
where $N_{v}=\binom{n}{v}$ is the number of sub-graphs with $v$
vertices. We are going to analyze $P_{bad}\left(v\right)$ as $n\rightarrow\infty$.
If for all $v$, $P_{bad}\left(v\right)=\textrm{o}\left(\frac{1}{n}\right)$,
then $\sum_{v=dk}^{n}P_{bad}\left(v\right)\le\textrm{o}\left(1\right)$
and the cuckoo hash succeeds with high probability. The analysis is
similar to the analysis given in \cite{Space efficient hash tables}.
Let $x_{0}=\exp\left(\frac{-2}{dk-2}\right)$. We divide the analysis
into two sections. In section \ref{sub:x<x0 Analysis} we show that
for any memory utilization $\beta$ and $\forall\frac{dk}{n}\le x<x_{0}$,
$P_{bad}\left(x\right)$ is $\textrm{o}\left(\frac{1}{n}\right)$.
In section \ref{sub:x>x0 Analysis}we find the maximum memory utilization
$\beta$ such that $\forall x_{0}\le x<\beta$, $P_{bad}\left(x\right)$
is exponentially small.

\subsubsection{\label{sub:x<x0 Analysis}\textup{$P_{bad}\left(x\right)$} Analysis
for $\frac{dk}{n}\le x<x_{0}$}

In this section we show that if $dk>2$ then for any memory utilization
$\beta$ and $\forall\frac{dk}{n}\le x<x_{0}$, $P_{bad}\left(x\right)$
is $\textrm{o}\left(\frac{1}{n}\right)$.
\begin{lem}
\label{lem:Pv<c0c1}$P_{bad}\left(x\right)<c_{0}\left(x\right)\cdot c_{1}^{n}\left(x\right)$, 

where $c_{0}\left(x\right)=\mathrm{e}\cdot x^{dk-1}$ and $c_{1}\left(x\right)=\mathrm{e}^{2x}\cdot x^{\left(dk-2\right)x}$\textup{.}\end{lem}
\begin{proof}
See appendix \ref{sec:Proof-of-Pv<c0c1}.\end{proof}
\begin{thm}
\label{lem:Pv is good t=00003Dn}

$\mbox{Let}$ $\delta$ be a small constant, $0<\delta\ll$1. If $dk\ge3$
then for any load $\beta$:

1. If $\frac{dk}{n}\le x\le\frac{dk}{n}+\delta$, then $P_{bad}\left(x\right)<\textrm{O\ensuremath{\left(n^{-\left(\left(dk\right)^{2}-dk-1\right)}\right)}}=\textrm{o}\left(\frac{1}{n}\right)$
. 

2. If $\frac{dk}{n}+\delta\le x\le x_{0}-\delta$, then $P_{bad}\left(x\right)$
decreases exponentially as $n\rightarrow\infty$.\end{thm}
\begin{proof}
For any memory utilization $\beta$, $P_{bad}\left(x\right)<c_{0}\left(x\right)\cdot c_{1}^{n}\left(x\right)$
,

where $c_{0}\left(x\right)=\mathrm{e}\cdot x^{dk-1}$ and $c_{1}\left(x\right)=\mathrm{e}^{2x}\cdot x^{\left(dk-2\right)x}$. 

Note that $c_{0}\left(x\right)$ and $c_{1}\left(x\right)$ are independent
of $\beta$. Recall that $x_{0}=\exp\left(\frac{-2}{dk-2}\right).$
$\forall\frac{dk}{n}+\delta\le x\le x_{0}-\delta$, there exist a
constant $\epsilon$ such that $c_{1}\left(x\right)<1-\epsilon$ .
We obtain that: 

1. If $x\rightarrow\frac{dk}{n}$, then $P_{bad}\left(x\right)<\textrm{lim}{}_{x\rightarrow\frac{dk}{n}}c_{0}\left(x\right)\cdot c_{1}^{n}\left(x\right)=\textrm{\textrm{O\ensuremath{\left(n^{-\left(\left(dk\right)^{2}-dk-1\right)}\right)}}}\le\textrm{O}\left(n^{-5}\right)=\textrm{o}\left(\frac{1}{n}\right)$
.

2. If $\frac{dk}{n}+\delta\leq x\le x_{0}-\delta,$ then $c_{1}\left(x\right)<1-\epsilon$,
and $P_{bad}\left(x\right)<c_{0}\left(x\right)\cdot c_{1}^{n}\left(x\right)$
decreases exponentially as $n\rightarrow\infty$.
\end{proof}

\subsubsection{\label{sub:x>x0 Analysis}\textup{$P_{bad}\left(x\right)$} Analysis
for $x_{0}\le x<\beta$}

In this section we are going to find the maximum memory utilization
$\beta$ such that $\forall x_{0}\le x<\beta$, the probability that
there exists a bad sub-graph is exponentially small.
\begin{lem}
\label{lem:Pv<c2c3c4c5}$\forall x_{0}\le x<\beta$, $P_{bad}\left(x\right)<\textrm{O}(1)\cdot c_{5}^{n}\left(x,\beta\right)$
where

$c_{5}\left(x,\beta\right)=\left(\frac{1}{1-x}\right)^{\left(1-x\right)}\left(\frac{1}{x}\right)^{x}\left(\frac{\beta}{\beta-x}\right)^{\left(\beta-x\right)}\left(\frac{\beta}{x}\right)^{x}x^{dkx}\left(1-dk\left(1-x\right)x^{dk-1}-x^{dk}\right)^{\beta-x}$\end{lem}
\begin{proof}
See appendix \ref{sec:Proof-Of-Pv<c2c3c4c5}\end{proof}
\begin{thm}
If $c_{5}\left(x,\beta\right)<1-\epsilon$ for all $x_{0}\le x<\beta$,
then $P_{bad}\left(x\right)$ decreases exponentially as $n\rightarrow\infty$.
Any memory utilization $\beta$ that satisfies the constraint is a
lower bound on the possible memory utilization.\end{thm}
\begin{proof}
The theorem follows directly from the inequality $P_{bad}\left(x\right)<\textrm{O}(1)\cdot c_{5}^{n}\left(x,\beta\right)$. 
\end{proof}
Numerical solutions to the constraint $c_{5}\left(x,\beta\right)<1$
indicate that the memory utilization of the CUCKOO-K algorithm is
$\beta_{Choose-2}(k=2,d=2,t=n)>0.937$ and $ $$\beta_{Choose-3}(k=3,d=2,t=n)>0.993$.
Our empirical results show that $\beta_{Choose-2}(k=2,d=2,t=n)=0.9768$
and $\beta_{Choose-3}(k=3,d=2,t=n)=0.9974$. The memory utilization
for $kd>6$ rapidly approaches one. Theoretical analysis performed
by \cite{The random graph threshold for k-orientiability,The k-orientability thresholds}
provided tight thresholds of the memory utilization of the CUCKOO-DISJOINT
algorithm. $\beta_{Disjoint}\left(k=2,d=2,t=n\right)=0.8970$ and
$\beta_{Disjoint}\left(k=3,d=2,t=n\right)=0.9592$. Ref. \cite{3.5 way}
do not provide a theoretical memory utilization threshold for small
$k$. The empirical results of Ref. \cite{3.5 way} show that in the
CUCKOO-OVERLAP algorithm $\beta_{Overlap}(k=2,d=2,t=n)=0.9650$ and
$\beta_{Overlap}(k=3,d=2,t=n)=0.9945$. The theoretical analysis of
the CUCKOO-DISJOINT algorithm performed in \cite{Maximum Matchings,Tight Thresholds for Cuckoo Hashing,Orientability of Random Hypergraphs}
does not apply for $k>1$ and the theoretical analysis of the CUCKOO-DISJOINT
algorithm performed by \cite{Multiple-orientability Thresholds} does
not apply for $d<3$.

\subsection{Memory utilization when the page size $t$ is a given constant}

In this section we analyze the probability that the hashing fails
for the case where the page size $t$ equals a constant. Let $P_{fail}\left(v\right)$
be the probability that there exists a sub-graph $S$ with $v$ vertices
such that $S$ has more edges than vertices and every vertex is in
at least one edge.

Let $x_{1}=\mbox{\ensuremath{\exp\left(\frac{-\left(k+1\right)}{dk-\left(k+1\right)}\right)}}$.
Here again we divide the analysis into two sections. In section \ref{sub:x<x1 Analysis}
we show that for any memory utilization $\beta$ and $\forall\frac{dk}{n}\le x<x_{1}$,
$P_{fail}\left(x\right)$ is $\textrm{o}\left(\frac{1}{n}\right)$.
In section \ref{sub:x>x1 Analysis} we find the maximum memory utilization
$\beta$ such that $\forall x_{1}\le x<\beta$, $P_{bad}\left(x\right)$
is exponentially small.

\subsubsection{\label{sub:x<x1 Analysis}$P_{fail}\left(x\right)$ Analysis for
$\frac{dk}{n}\le x<x_{1}$}

In this section we show that for any memory utilization $\beta$ and
$\forall\frac{dk}{n}\le x<x_{1}$, $P_{fail}\left(x\right)$ is $\textrm{o}\left(\frac{1}{n}\right)$.

The analysis here is similar to the case above where $t=n$, however
now we need to take into consideration the distribution of the vertices
over the pages. 

Let $V$ be a given set of vertices, and let $v_{i}$ be the number
of vertices in page $i$. The probability that a random edge hits
$V$ is equal to $p_{hit}\left(V,t\right)=\left(\sum_{i=1}^{g}\frac{\binom{v_{i}}{k}}{g\cdot\binom{t}{k}}\right)^{d}\leq\left(\frac{1}{g}\sum_{i=1}^{g}\left(\frac{v_{i}}{t}\right)^{k}\right)^{d}$.

Let $\widetilde{p}_{hit}\left(V,t\right)=\left(\frac{1}{g}\sum_{i=1}^{g}\left(\frac{v_{i}}{t}\right)^{k}\right)^{d}$
be an upper bound on $p_{hit}\left(V,t\right)$.

We are going to use the following lemma which states that increasing
the page size reduces $\widetilde{p}_{hit}$.
\begin{lem}
\label{lem: low t is bad} For any set of vertices V and any integer
c,$\widetilde{p}_{hit}\left(V,t\right)\ge\widetilde{p}_{hit}\left(V,ct\right)$. \end{lem}
\begin{proof}
Since the function $f\left(v\right)=\left(\frac{v}{t}\right)^{k}$
is convex, for any sequence of $c$ pages, $\frac{1}{c}\left(\left(\frac{v_{1}}{t}\right)^{k}+\ldots+\left(\frac{v_{c}}{t}\right)^{k}\right)\ge\left(\frac{v_{1}+\ldots+v_{c}}{ct}\right)^{k}$,
thus multiplying the page size by a factor $c$ does not increase
$\widetilde{p}_{hit}$. For convenience, we restrict our analysis
to pages of size $t=k\cdot c$, where c is any integer. The worst
case is obtained when $t=k$ which is equivalent to the classical
CUCKOO-DISJOINT hashing.

We will now analyze $P_{fail}\left(x\right)$ and $P_{bad}\left(x\right)$
as $n\rightarrow\infty$. Recall that $\beta=\frac{m}{n}$ $ $ is
the memory utilization, $0<\beta\leq1$ and $x=\frac{v}{n}$. \end{proof}
\begin{lem}
\label{lem:Pv<c6c7}$P_{fail}\left(x\right)<c_{6}\left(x\right)\cdot c_{7}^{n}\left(x\right)$ 

where $c_{6}\left(x\right)=\mathrm{e}x^{d-1}$ and $c_{7}\left(x\right)=\mathrm{e}^{\frac{\left(k+1\right)x}{k}}x^{\frac{\left(dk-1-k\right)x}{k}}$.\end{lem}
\begin{proof}
See appendix \ref{sec:Proof-of Pv<c6c7}.\end{proof}
\begin{thm}
$\mbox{Let}$ $\delta$ be a small constant, $0<\delta\ll$1. If $\left(d-1\right)dk\ge3$,
then for any load $\beta$:

1. If $\frac{dk}{n}\le x\le\frac{dk}{n}+\delta$, then $P_{fail}\left(x\right)<\textrm{O\ensuremath{\left(n^{-\left(\left(d-1\right)dk-1\right)}\right)}}=\textrm{o}\left(\frac{1}{n}\right)$
. 

2. If $\frac{dk}{n}+\delta\le x\le x_{1}-\delta$, then $P_{fail}\left(x\right)$
decreases exponentially as $n\rightarrow\infty$.\end{thm}
\begin{proof}
For any memory utilization $\beta$, $P_{fail}\left(x\right)<c_{6}\left(x\right)\cdot c_{7}^{n}\left(x\right)$ 

where $c_{6}\left(x\right)=\mathrm{e}x^{d-1}$ and $c_{7}\left(x\right)=\mathrm{e}^{\frac{\left(k+1\right)x}{k}}x^{\frac{\left(dk-1-k\right)x}{k}}$.

Note that $c_{6}\left(x\right)$ and $c_{7}\left(x\right)$ are independent
of $\beta$. Recall that $x_{1}=\mbox{\ensuremath{\exp\left(\frac{-\left(k+1\right)}{dk-\left(k+1\right)}\right)}}$.
$\forall\frac{dk}{n}+\delta\le x\le x_{1}-\delta$, there exist a
constant $\epsilon$ such that $c_{7}\left(x\right)<1-\epsilon$ .
We obtain that: 

1. If $x\rightarrow\frac{dk}{n}$, then $P_{fail}\left(x\right)<\textrm{lim}{}_{x\rightarrow\frac{dk}{n}}c_{6}\left(x\right)\cdot c_{7}^{n}\left(x\right)=\textrm{O}\left(n^{-\left(\left(d-1\right)dk-1\right)}\right)\le\textrm{O}\left(n^{-2}\right)=\textrm{o}\left(\frac{1}{n}\right)$
.

2. If $\frac{dk}{n}+\delta\leq x<x_{1}-\delta$, then $c_{7}\left(x\right)<1-\epsilon$,
and $P_{fail}\left(x\right)<c_{6}\left(x\right)\cdot c_{7}^{n}\left(x\right)$
decreases exponentially as $n\rightarrow\infty$.
\end{proof}

\subsubsection{\label{sub:x>x1 Analysis}$P_{bad}\left(x\right)$ Analysis for $x_{1}\le x<\beta$}

In this section we find the maximum memory utilization $\beta$ such
that $\forall x_{1}<x<\beta$, the probability that there exists a
bad sub-graph is exponentially small. We examine the set of sub-graphs
that have a given distribution $\underline{\text{a}}$ of vertices
over the pages. $\underline{\text{a}}=\left(a_{0},...,a_{t}\right)$,
where $a_{i}$ is the number of pages that have $i$ vertices. For
example, when the page size $t$ was equal to $n$ and the number
of pages $g$ was 1, the number of sub-graphs with $v$ vertices was
$\binom{n}{v}$ and the corresponding $\underline{\text{a}}$ of those
sub-graphs was $\left\{ 0,...,0,a_{v}=1,0,...,0\right\} $. by definition
of $\underline{\text{a}}$,

\begin{equation}
\sum_{i=0}^{t}a_{i}=g\end{equation}

\begin{lem}
When the page size $t$ is a given constant, the number of different
possible values of $\underline{\text{a}}$ is polynomial in $n$.\end{lem}
\begin{proof}
We denote by $\#\underline{\text{a}}$ The number of different possible
values of $\underline{\text{a}}$.

$\#\underline{\text{a}}<g^{t}=\left(\frac{n}{t}\right)^{t}$. Since
$t$ is constant, $\#a$ is polynomial in $n$.
\end{proof}
Let $P_{bad}\left(\underline{\text{a}}\right)$ be the probability
that there exists a bad sub-graph with distribution $\underline{\text{a}}$
of vertices over the pages. If $P_{bad}\left(\underline{\text{a}}\right)$
is exponentially small for every $\underline{\text{a}}$, then the
union bound over a polynomial number of all possible values of $a$
is also exponentially small. Let $p_{bad}\left(\underline{\text{a}}\right)$
be the probability that \emph{a given }sub-graph $S\left(V,E\right)$
with $\underline{\text{a}}=\left(a_{0},...,a_{t}\right)$ is bad.

Let $\hat{a}=\frac{\underline{a}}{g}$ be a unit vector. When the
page size $t$ is a constant, the probability that a random edge hits
$V$ is:

\begin{equation}
p_{hit}\left(\underline{\text{a}}\right)=\left(\frac{\sum_{i=k}^{t}a_{i}\binom{i}{k}}{g\binom{t}{k}}\right)^{d}=\left(\frac{\sum_{i=k}^{t}\hat{a}_{i}\binom{i}{k}}{\binom{t}{k}}\right)^{d}\end{equation}
 and the probability that a random edge connects $V$ to exactly one
vertex from outside of V is

\begin{equation}
p_{1}\left(\underline{\text{a}}\right)=d\left(\frac{\sum_{i=k}^{t}a_{i}\left(t-i\right)\binom{i}{k-1}}{g\binom{t}{k}}\right)\left(\frac{\sum_{i=k}^{t}a_{i}\binom{i}{k}}{g\binom{t}{k}}\right)^{d-1}=d\left(\frac{\sum_{i=k}^{t}\hat{a}_{i}\left(t-i\right)\binom{i}{k-1}}{\binom{t}{k}}\right)\left(\frac{\sum_{i=k}^{t}\hat{a}_{i}\binom{i}{k}}{\binom{t}{k}}\right)^{d-1}\end{equation}
 Using the union bound we obtain that

\begin{equation}
P_{bad}\left(\underline{\text{a}}\right)<N\left(\underline{\text{a}}\right)\cdot p_{bad}\left(\underline{\text{a}}\right)\end{equation}
 where $N\left(\underline{\text{a}}\right)$ is the number of sub-graphs
with $\underline{\text{a}}=\left(a_{0},...,a_{t}\right)$.

\begin{equation}
N\left(\underline{\text{a}}\right)=\binom{g}{a_{0},...,a_{t}}\prod_{i=0}^{t}\binom{t}{_{i}}^{a_{i}}\end{equation}
For the asymptotic behavior of $N\left(\underline{\text{a}}\right)$,
we are going to use the following lemma:
\begin{lem}
\label{lem:multinom_bound}$\binom{g}{a_{0},...,a_{t}}\leq\prod_{i=0}^{t}\left(\frac{g}{a_{i}}\right)^{a_{i}}$\end{lem}
\begin{proof}
See Appendix \ref{sec:Proof-of-Lemma multinom_bound}. \end{proof}
\begin{lem}
\label{lem:Pv<c8c9}$P_{bad}\left(\hat{a},\beta\right)<\textrm{O}(1)\cdot c_{8}\left(\hat{a},\beta\right)\cdot c_{9}^{n}\left(\hat{a},\beta\right)$
where

$c_{8}\left(\hat{a},\beta\right)=p_{hit}\left(1-p_{1}-p_{hit}\right)^{-1}\left(\frac{\beta-x}{x}\right)$
and

$c_{9}\left(\hat{a},\beta\right)=\prod_{i=0}^{t}\left(\frac{1}{\hat{a}_{i}}\right)^{\frac{\hat{a}_{i}}{t}}\prod_{i=0}^{t}\binom{t}{_{i}}^{\frac{\hat{a}_{i}}{t}}\left(\frac{\beta}{\beta-x}\right)^{\left(\beta-x\right)}\left(\frac{\beta}{x}\right)^{x}p_{hit}^{x}\left(1-p_{1}-p_{hit}\right)^{\beta-x}$\end{lem}
\begin{proof}
See appendix \ref{sec:Proof-of-Pv<c8c9}.\end{proof}
\begin{thm}
If $ $$c_{9}\left(\hat{a},\beta\right)<1-\epsilon$ for all $\forall\, x_{1}\le x<\beta$,
then $P_{bad}\left(\hat{a},\beta\right)$ decreases exponentially
as $n\rightarrow\infty$ . Any memory utilization $\beta$ that satisfies
the constraint is a lower bound on the possible memory utilization.\end{thm}
\begin{proof}
The theorem follows directly from the inequality $P_{bad}\left(\hat{a},\beta\right)<\textrm{O}(1)\cdot c_{8}\left(\hat{a},\beta\right)\cdot c_{9}^{n}\left(\hat{a},\beta\right)$.
\end{proof}
Theoretical lower bounds of the memory utilization obtained from numerical
solutions of the constraint $c_{9}\left(\hat{a},\beta\right)<1$ are
displayed in figure \ref{fig:Memory-utilization Vs. page size}.

\section{Empirical Results}

\begin{figure}
\includegraphics[scale=0.5]{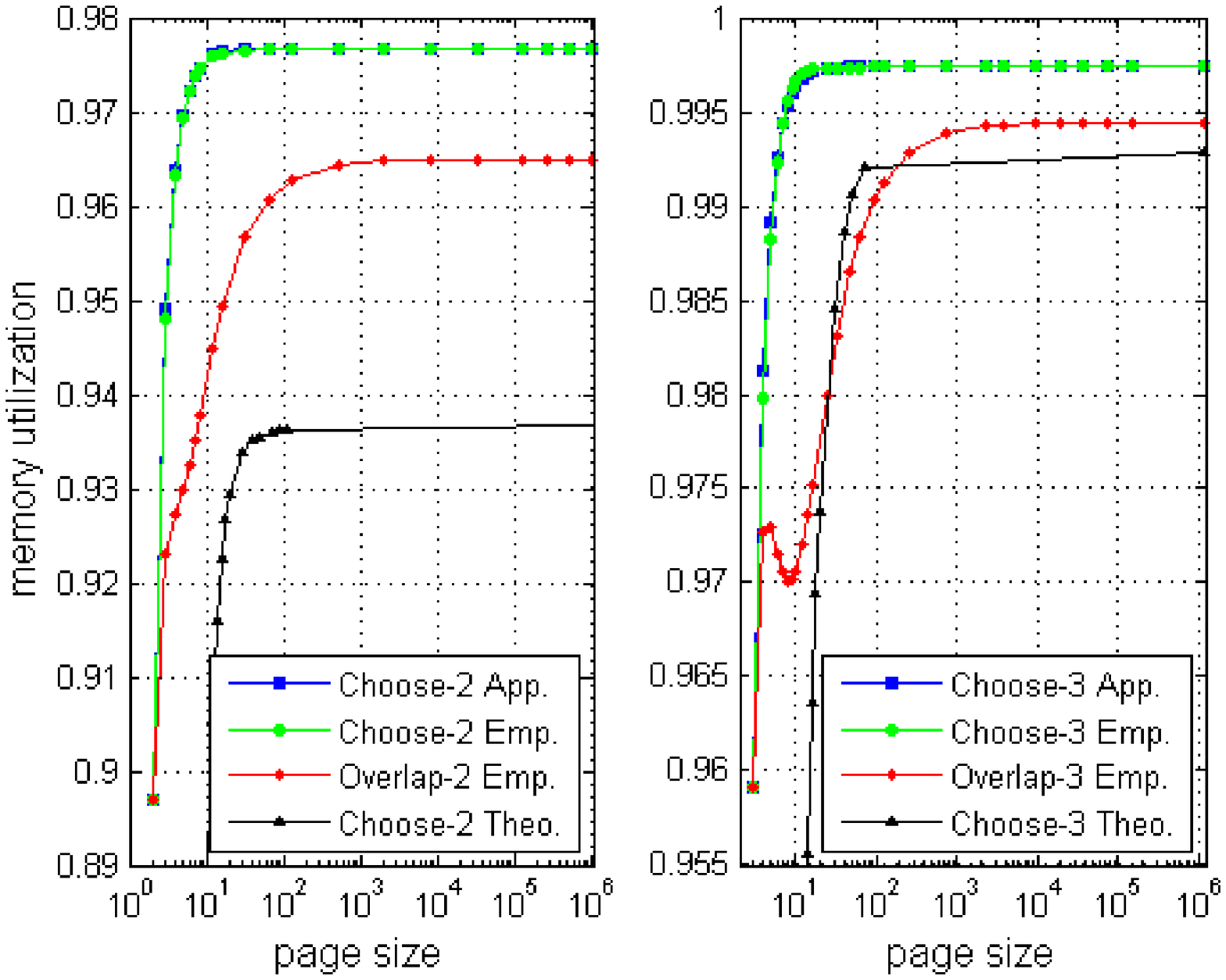}

\caption{\label{fig:Memory-utilization Vs. page size}}

Memory utilization vs. page size. Empirical CUCKOO-CHOOSE-K (green),
Empirical CUCKOO-OVERLAP (red), approximation formula of Empirical
CUCKOO-CHOOSE-K (blue), and theoretical lower bound of CUCKOO-CHOOSE-K
(black). left: $k=2$, right: $k=3$.$ $ 
\end{figure}

The experiments were conducted with a similar protocol to the one
described in \cite{3.5 way}. In all experiments the number of the
buckets, $d$, was two. The capacity of each bucket $k$ was either
two or three. The size of the hash tables $n$ was $1,209,600$. The
reported memory utilization $\beta$ is the mean memory utilization
over twenty trials. The random hash functions were based on the Matlab
{}``rand'' function with the twister method. Items were inserted
into the hash table one-by-one until an item could not be inserted.
The results of both CUCKOO-CHOOSE-K and CUCKOO-OVERLAP were notably
stable. In each case, the standard deviation was a few hundredths
of a percent, so error bars would be invisible in the figure. Such
strongly predictable behavior is appealing from a practical standpoint.
Since we added a paging constraint, our results are not comparable
to previous works that do not include a paging constraint.

\begin{figure}
\includegraphics[scale=0.4]{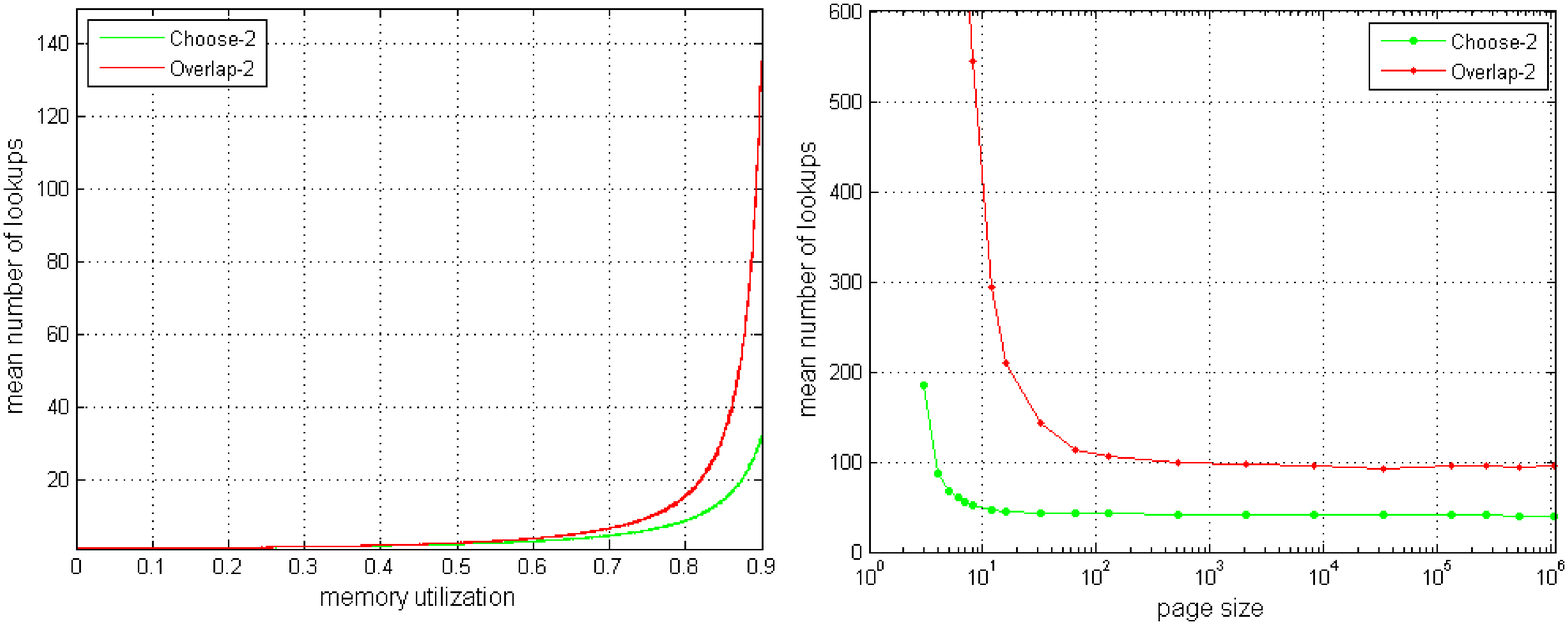}\caption{\label{fig:Number-of-iteration}}

Number of lookups required to insert an item vs. memory utilization
(left) and vs. page size (right). In the left figure the page size
is 8. In the right figure the memory utilization is 92\%. The left
figure was smoothed with an averaging filter. The variance in the
number of lookups required to insert an item was much smallar in CUCKOO-CHOOSE-K.
\end{figure}

Experiments show that CUCKOO-CHOOSE-K improves memory utilization
significantly even for a small page size $t$ and a small bucket capacity
$k$ when compared to the classical cuckoo hashing CUCKOO-DISJOINT.
It outperforms CUCKOO-OVERLAP as well. Recall that CUCKOO-DISJOINT
is the extreme case of CUCKOO-CHOOSE-K and CUCKOO-OVERLAP when the
page size is equal to bucket size $k$. The memory utilization $\beta_{Choose-k}$
converges very quickly to its maximum value. For example $\beta_{Choose-2}(t=16)=0.9763$
is almost equal to $\beta_{Choose-2}(t=1,209,600)=0.9767$. Whereas
$\beta_{Overlap-2}\left(t=16\right)=0.9494$ and $\beta_{Disjoint-2}=0.8970$.

The empirical memory utilization can be approximated very accurately
by the following formulas:

\begin{equation}
\beta_{Choose-2}\left(t\right)\approx0.977\cdot\left(1-(\frac{t}{0.764})^{-2.604}\right)\end{equation}

\begin{equation}
\beta_{Choose-3}\left(t\right)\approx0.997\cdot\left(1-(\frac{t}{1.011})^{-2.998}\right)\end{equation}

The maximum approximation error is 0.0011 for $k=2$ and 0.0015 for
$k=3$.

The empirical results and their approximations are displayed in figure
\ref{fig:Memory-utilization Vs. page size} together with empirical
results of CUCKOO-OVERLAP. The memory utilization for $kd>6$ rapidly
approaches one (not displayed).

CUCKOO-CHOOSE-K outperforms CUCKOO-OVERLAP not only in memory utilization,
but in the number of iterations required to insert a new item as well.
Figure \ref{fig:Number-of-iteration} illustrates the number of iterations
required to insert a new item when the hash table is 92\% full and
we use breadth first search to search for a vacant position. $\#itr_{Choose-2}(t=8)=$52,
whereas $\#itr_{Overlap-2}(t=8)=545$. Note that in these simulations
we did not limit the number of insert iterations and we continued
to insert items as long as we could find free locations. Most applications
limit the number of inserted iterations, and maintain a low memory
utilization in order to find a free location easily. If an empty position
is not found within a fixed number of iterations, a rehash is performed
or the item is placed outside of the cuckoo array.

\appendix

\section{\label{sec:Proof-of-Pv<c0c1}Proof of Lemma \ref{lem:Pv<c0c1}}

Here we prove that $ $$P_{bad}\left(x\right)<c_{0}\left(x\right)\cdot\left(c_{1}\left(x\right)\right)^{n}$,
where $c_{0}\left(x\right)=\mathrm{e}\cdot x^{dk-1}$ and $c_{1}\left(x\right)=\mathrm{e}^{2x}\cdot x^{\left(dk-2\right)x}$.
\begin{proof}
Recall that $\beta=\frac{m}{n}$ $ $ is the memory utilization, $0<\beta\leq1$
and $x=\frac{v}{n},$ $\frac{dk}{n}\le x<\beta$. (if $v\geq m$ or
$v<dk,$ then the sub-graph $S$ cannot fail).

$P_{bad}\left(v\right)<\binom{n}{v}\cdot p_{bad}\left(v\right)$ where

$p_{bad}\left(v\right)=\binom{m}{v+1}p_{hit}^{v+1}\left(1-p_{1}-p_{hit}\right)^{m-\left(v+1\right)},$

$p_{hit}\left(v\right)=\left(\frac{\binom{v}{k}}{\binom{n}{k}}\right)^{d}$,
and

$p_{1}=d\frac{\left(n-v\right)\binom{v}{k-1}}{\binom{n}{k}}\left(\frac{\binom{v}{k}}{\binom{n}{k}}\right)^{d-1}$
.

$\binom{n}{v},$ $p_{bad}$, $\binom{m}{v+1}$ and $p_{hit}$ are
bounded by:

\begin{equation}
\binom{n}{v}<\left(\frac{\mathrm{e}\cdot n}{v}\right)^{v}=\left(\frac{\mathrm{e}}{x}\right)^{xn}.\end{equation}

\begin{equation}
p_{bad}\left(v\right)=\binom{m}{v+1}p_{hit}^{v+1}\left(1-p_{1}-p_{hit}\right)^{m-\left(v+1\right)}<\binom{m}{v+1}p_{hit}^{v+1}.\end{equation}

\begin{equation}
\binom{m}{v+1}<\left(\frac{\mathrm{e}\cdot m}{v+1}\right)^{v+1}<\left(\frac{\mathrm{e}\cdot m}{v}\right)^{v+1}=\left(\frac{\mathrm{e}\cdot\beta}{x}\right)^{xn+1}<\left(\frac{\mathrm{e}}{x}\right)^{xn+1}.\end{equation}

\begin{equation}
p_{hit}<\left(\frac{v}{n}\right)^{dk}=x^{dk}.\end{equation}

And we get that

\begin{equation}
P_{bad}\left(v\right)<\binom{n}{v}\cdot p_{bad}\left(v\right)<\left(\frac{\mathrm{e}}{x}\right)^{xn}\left(\frac{\mathrm{e}}{x}\right)^{xn+1}\left(x^{dk}\right)^{xn+1}=c_{0}\left(x\right)\cdot c_{1}^{n}\left(x\right).\end{equation}

\end{proof}

\section{\label{sec:Proof-Of-Pv<c2c3c4c5}Proof Of Lemma \ref{lem:Pv<c2c3c4c5}}

Here we prove $\forall x_{0}\le x<\beta$, $P_{bad}\left(x,\beta\right)<\textrm{O}(1)\cdot c_{5}^{n}\left(x,\beta\right)$
where

$c_{5}\left(x,\beta\right)=\left(\frac{1}{1-x}\right)^{\left(1-x\right)}\left(\frac{1}{x}\right)^{x}\left(\frac{\beta}{\beta-x}\right)^{\left(\beta-x\right)}\left(\frac{\beta}{x}\right)^{x}x^{dkx}\left(1-dk\left(1-x\right)x^{dk-1}-x^{dk}\right)^{\beta-x}$.
\begin{proof}
Recall that: \begin{equation}
P_{bad}\left(v\right)<N\left(v\right)\cdot p_{bad}\left(v\right)=\binom{n}{v}\binom{m}{v+1}p_{hit}^{v+1}\left(1-p_{1}-p_{hit}\right)^{m-\left(v+1\right)}.\end{equation}
$\binom{n}{v}$, $p_{hit}$ and $p_{1}$ are bounded by:$ $

\begin{equation}
\binom{n}{v}<\left(\frac{n}{n-v}\right)^{n-v}\left(\frac{n}{v}\right)^{v}=\left(\frac{1}{1-x}\right)^{\left(1-x\right)n}\left(\frac{1}{x}\right)^{xn},\end{equation}

\begin{equation}
\left(x-\frac{k}{n}\right)^{dk}<p_{hit}<x^{dk},\end{equation}

\begin{equation}
dk\frac{\left(1-x\right)}{x}\left(x-\frac{k}{n}\right)^{dk}<p_{1}.\end{equation}

Since $x\ge x_{0}\gg\frac{1}{n},\frac{k}{n}$, we neglect the term
$\frac{1}{n}\ll x$ in the following approximation of $\binom{m}{v+1}$,
and we neglect the term $\frac{k}{n}\ll x$ in the lower bounds for
$p_{hit}$ and $p_{1}$. As $n\rightarrow\infty$, these terms contribute
to $P_{bad}\left(v\right)$ factors which are bounded by $\textrm{O}(1)$$ $.

\begin{equation}
\binom{m}{v+1}<\left(\frac{m}{m-\left(v+1\right)}\right)^{m-\left(v+1\right)}\left(\frac{m}{v+1}\right)^{v+1}<\textrm{O}(1)\cdot\left(\frac{\beta}{\beta-x}\right)^{\left(\beta-x\right)n-1}\left(\frac{\beta}{x}\right)^{xn+1}\end{equation}

\begin{equation}
\left(1-p_{1}-p_{hit}\right)^{m-\left(v+1\right)}<\textrm{O}(1)\cdot\left(1-dk\left(1-x\right)x^{dk-1}-x^{dk}\right)^{\beta n-\left(xn+1\right)}\end{equation}
 The proof for the left inequalities is given in \cite{Space efficient hash tables}
and is also a special case of the more general inequality we prove
later in lemma \ref{lem:multinom_bound}.

$ $We get that:

\begin{equation}
P_{bad}\left(v\right)<\binom{n}{v}\cdot p_{bad}\left(v\right)<\textrm{O}(1)\cdot c_{4}\left(x,\beta\right)\cdot c_{5}^{n}\left(x,\beta\right),\end{equation}
where

\begin{equation}
c_{4}\left(x,\beta\right)=\left(\beta-x\right)x^{dk-1}\left(1-dk\left(1-x\right)x^{dk-1}-x^{dk}\right)^{-1}\end{equation}
 and

\begin{equation}
c_{5}\left(x,\beta\right)=\left(\frac{1}{1-x}\right)^{\left(1-x\right)}\left(\frac{1}{x}\right)^{x}\left(\frac{\beta}{\beta-x}\right)^{\left(\beta-x\right)}\left(\frac{\beta}{x}\right)^{x}x^{dkx}\left(1-dk\left(1-x\right)x^{dk-1}-x^{dk}\right)^{\beta-x}\end{equation}
Since $c_{4}<\textrm{O}\left(1\right),$ we get: $P_{bad}\left(x\right)<\textrm{O}(1)\cdot c_{5}^{n}\left(x,\beta\right)$.
\end{proof}

\section{\label{sec:Proof-of Pv<c6c7}Proof of Lemma \ref{lem:Pv<c6c7}}

Here we prove $P_{fail}\left(v\right)<c_{6}\left(x\right)\cdot c_{7}^{n}\left(x\right)$ 

where $c_{6}\left(x\right)=\mathrm{e}x^{d-1}$ and $c_{7}\left(x\right)=\mathrm{e}^{\frac{\left(k+1\right)x}{k}}x^{\frac{\left(dk-1-k\right)x}{k}}$.
\begin{proof}
According to lemma \ref{lem: low t is bad}, for any set of vertices
V, $\widetilde{p}_{hit}$ decreases when the page size $t$ is multiplied
by an integer. For simplicity, we restrict our analysis to pages of
size $t=k\cdot c$, where $c$ is any integer. The worst case therefore
is when $t=k$ which is equivalent to the classical CUCKOO-DISJOINT
hashing. We use the union bound to obtain $P_{fail}\left(v\right)<N\left(v\right)\cdot p_{fail}\left(v\right)$,
where $N\left(v\right)$ is the number of sub-graphs with $v$ vertices,
where each vertex of the sub-graph is hit by at least one edge, and
$p_{fail}\left(v\right)$ is the probability that a given\emph{ }sub-graph
$S\left(V,E\right)$ was hit by more than $v$ edges. By definition
of $p_{fail}$:

\begin{equation}
p_{fail}\left(v\right)<\binom{m}{v+1}\widetilde{p}_{hit}^{v+1}.\end{equation}
when $t=k$ we get:

\begin{equation}
\widetilde{p}{}_{hit}^{v+1}=\left(x^{d}\right)^{xn+1}=\left(x^{d}\right)\left(x^{dx}\right)^{n},\end{equation}
and

\begin{equation}
N\left(v\right)\le\binom{n/k}{v/k}<\left(\frac{\mathrm{e}n/k}{v/k}\right)^{v/k}=\left(\left(\frac{\mathrm{e}}{x}\right)^{x/k}\right)^{n},\end{equation}
Since when the page size $t$ equals $k$, an element can be placed
either in all of the locations of a page or in none of them.

\begin{equation}
\binom{m}{v+1}<\binom{n}{v+1}<\left(\frac{\mathrm{e}n}{v+1}\right)^{v+1}<\left(\frac{\mathrm{e}n}{v}\right)^{v+1}=\left(\frac{\mathrm{e}}{x}\right)\left(\left(\frac{\mathrm{e}}{x}\right)^{x}\right)^{n}.\end{equation}
We get that:

\begin{equation}
P_{fail}\left(v\right)<N\left(v\right)\cdot p_{fail}\left(v\right)<N\left(v\right)\binom{m}{v+1}\widetilde{p}_{hit}^{v+1}<c_{6}\left(x\right)\cdot c_{7}^{n}\left(x\right),\end{equation}
where $c_{6}\left(x\right)=\mathrm{e}x^{d-1}$ and $c_{7}\left(x\right)=\mathrm{e}^{\frac{\left(k+1\right)x}{k}}x^{\frac{\left(dk-1-k\right)x}{k}}$.
\end{proof}

\section{\label{sec:Proof-of-Lemma multinom_bound}Proof of Lemma \ref{lem:multinom_bound}}

The proof of $\binom{g}{a_{0},...,a_{t}}\leq\prod_{i=0}^{t}\left(\frac{g}{a_{i}}\right)^{a_{i}}$
is a generalization of the proof given in \cite{Space efficient hash tables}.
For any positive integer $t$ and any non-negative integer $g$: $\left(y_{0}+...+y_{t}\right)^{g}=\sum\limits _{\alpha_{0}+...+\alpha_{t}=g}\binom{g}{\alpha_{0},...,\alpha_{t}}\prod_{i=0}^{t}\left(y_{i}^{\alpha_{i}}\right)$,
where the summation is taken over all sequences of non-negative integer
indices $\alpha_{0}$ through $\alpha_{t}$ such that the sum of all
$\alpha_{i}$ is $g$. For the special case where $y_{i}=\frac{a_{i}}{g}$,
$a_{i}$ is a non-negative integer and $\sum_{i=0}^{t}a_{i}=g$ we
get:

\begin{equation}
\binom{g}{a_{0},...,a_{t}}\prod_{i=0}^{t}\left(\frac{a_{i}}{g}\right)^{a_{i}}\leq\sum\limits _{\alpha{}_{0}+...+\alpha{}_{t}=g}\binom{g}{\alpha_{0},...,\alpha{}_{t}}\prod_{i=0}^{t}\left(\frac{a_{i}}{g}\right)^{\alpha{}_{i}}=\left(y_{0}+...+y_{t}\right)^{g}=1\end{equation}

So \begin{equation}
\binom{g}{a_{0},...,a_{t}}\leq\frac{1}{\prod_{i=0}^{t}\left(\frac{a_{i}}{g}\right)^{a_{i}}}=\prod_{i=0}^{t}\left(\frac{g}{a_{i}}\right)^{a_{i}}\end{equation}

\section{\label{sec:Proof-of-Pv<c8c9}Proof of Lemma \ref{lem:Pv<c8c9}}

Here we prove

$P_{bad}\left(\hat{a},\beta\right)<\textrm{O}(1)\cdot c_{8}\left(\hat{a},\beta\right)\cdot c_{9}^{n}\left(\hat{a},\beta\right)$
where

$c_{8}\left(\hat{a},\beta\right)=p_{hit}\left(1-p_{1}-p_{hit}\right)^{-1}\left(\frac{\beta-x}{x}\right)$
and

$c_{9}\left(\hat{a},\beta\right)=\prod_{i=0}^{t}\left(\frac{1}{\hat{a}_{i}}\right)^{\frac{\hat{a}_{i}}{t}}\prod_{i=0}^{t}\binom{t}{_{i}}^{\frac{\hat{a}_{i}}{t}}\left(\frac{\beta}{\beta-x}\right)^{\left(\beta-x\right)}\left(\frac{\beta}{x}\right)^{x}p_{hit}^{x}\left(1-p_{1}-p_{hit}\right)^{\beta-x}.$
\begin{proof}
Recall that $\hat{a}$ is a unit vector and $v$ and $x$ are equal
to the following functions of $\hat{a}$:

\begin{equation}
\sum_{i=0}^{t}\hat{a}_{i}=1\end{equation}
 \begin{equation}
v=\sum_{i=0}^{t}a_{i}\cdot i=g\sum_{i=0}^{t}\hat{a}_{i}\cdot i=\frac{n}{t}\sum_{i=0}^{t}\hat{a}_{i}\cdot i\end{equation}

\begin{equation}
x=\frac{v}{n}=\frac{1}{t}\sum_{i=0}^{t}\hat{a}_{i}\cdot i.\end{equation}
$P_{bad}\left(v\right)<N\left(v\right)\cdot p_{bad}\left(v\right)$,
where

\begin{equation}
N\left(v\right)=\binom{g}{a_{0},...,a_{t}}\prod_{i=0}^{t}\binom{t}{_{i}}^{a_{i}},\end{equation}

\begin{equation}
p_{bad}\left(v\right)=\binom{m}{v+1}p_{hit}^{v+1}\left(1-p_{1}-p_{hit}\right)^{m-\left(v+1\right)},\end{equation}

\begin{equation}
p_{hit}=\left(\frac{\sum_{i=k}^{t}a_{i}\binom{i}{k}}{g\binom{t}{k}}\right)^{d}=\left(\frac{\sum_{i=k}^{t}\hat{a}_{i}\binom{i}{k}}{\binom{t}{k}}\right)^{d},\end{equation}

\begin{equation}
p_{1}=d\left(\frac{\sum_{i=k}^{t}\hat{a}_{i}\left(t-i\right)\binom{i}{k-1}}{\binom{t}{k}}\right)\left(\frac{\sum_{i=k}^{t}\hat{a}_{i}\binom{i}{k}}{\binom{t}{k}}\right)^{d-1}.\end{equation}
Using lemma \ref{lem:multinom_bound} we get:

\begin{equation}
N\left(v\right)=\binom{g}{a_{0},...,a_{t}}\prod_{i=0}^{t}\binom{t}{_{i}}^{a_{i}}\leq\left(\prod_{i=0}^{t}\left(\frac{1}{\hat{a_{i}}}\right)^{\frac{\hat{a}_{i}}{t}}\prod_{i=0}^{t}\binom{t}{_{i}}^{\frac{\hat{a}_{i}}{t}}\right)^{n},\end{equation}
 and

$ $

\begin{equation}
\binom{m}{v+1}<\left(\frac{m}{m-\left(v+1\right)}\right)^{m-\left(v+1\right)}\left(\frac{m}{v+1}\right)^{v+1}<\textrm{O}(1)\cdot\left(\frac{\beta}{\beta-x}\right)^{\left(\beta-x\right)n-1}\left(\frac{\beta}{x}\right)^{xn+1}.\end{equation}
Finally we get:\begin{equation}
P_{bad}\left(v\right)<N\left(v\right)\cdot p_{bad}\left(v\right)=\textrm{O}(1)\cdot c_{8}\left(x,\beta\right)\cdot c_{9}^{n}\left(\hat{a},\beta\right)\end{equation}
 where

\begin{equation}
c_{8}\left(\hat{a},\beta\right)=p_{hit}\left(1-p_{1}-p_{hit}\right)^{-1}\left(\frac{\beta-x}{x}\right),\end{equation}
 and

\begin{equation}
c_{9}\left(\hat{a},\beta\right)=\prod_{i=0}^{t}\left(\frac{1}{\hat{a}_{i}}\right)^{\frac{\hat{a}_{i}}{t}}\prod_{i=0}^{t}\binom{t}{_{i}}^{\frac{\hat{a}_{i}}{t}}\left(\frac{\beta}{\beta-x}\right)^{\left(\beta-x\right)}\left(\frac{\beta}{x}\right)^{x}p_{hit}^{x}\left(1-p_{1}-p_{hit}\right)^{\beta-x}.\end{equation}
\end{proof}

\end{document}